\documentclass[conference,onecolumn]{IEEEtran}
\IEEEoverridecommandlockouts

\usepackage{cite}
\usepackage{algorithmic}
\usepackage{graphicx}
\usepackage{textcomp}
\usepackage{xcolor}
\usepackage{balance}
\usepackage[cmex10]{amsmath}
\usepackage{amsmath}
\usepackage{amssymb}
\usepackage{amsthm}
\usepackage{amsfonts}
\usepackage{bm}
\usepackage{xfrac}
\usepackage{graphicx}
\usepackage[font=footnotesize]{caption}
\usepackage{subcaption} 
\usepackage{empheq}
\usepackage[normalem]{ulem} 
\usepackage{soul} 
\usepackage{mathtools}
\newtheorem{theorem}{Theorem}

\newtheorem{example}{Example}
\newtheorem{proposition}{Proposition}
\newtheorem{lemma}{Lemma}

\newtheorem{corollary}{Corollary}
\newtheorem{remark}{Remark}
\theoremstyle{definition}

\def\BibTeX{{\rm B\kern-.05em{\sc i\kern-.025em b}\kern-.08em
    T\kern-.1667em\lower.7ex\hbox{E}\kern-.125emX}}
\begin{document}

\title{Sparse Point-wise Privacy Leakage: Mechanism Design and Fundamental Limits 
}

\author{\IEEEauthorblockN{1\textsuperscript{st} Amirreza Zamani}
\IEEEauthorblockA{\textit{Department of Information Science and Engineering, KTH} \\
Stockholm, Sweden \\
amizam@kth.se}
\and
\IEEEauthorblockN{2\textsuperscript{nd} Sajad Daei}
\IEEEauthorblockA{\textit{Department of Information Science and Engineering, KTH} \\
Stockholm, Sweden \\
sajado@kth.se}
\and
 \IEEEauthorblockN{3\textsuperscript{rd} Parastoo Sadeghi}
\IEEEauthorblockA{\textit{School of Engineering and Technology, UNSW} \\
Canberra, Australia \\
p.sadeghi@unsw.edu.au}
\and
\IEEEauthorblockN{4\textsuperscript{th} Mikael Skoglund}
\IEEEauthorblockA{\textit{Department of Information Science and Engineering, KTH} \\
Stockholm, Sweden \\
skoglund@kth.se}}
\maketitle
\begin{abstract}

We study an information-theoretic privacy mechanism design problem, where an agent observes useful data $Y$ that is
arbitrarily correlated with sensitive data $X$, and design disclosed data $U$ generated from $Y$ (the agent
has no direct access to $X$). We introduce \emph{sparse point-wise privacy leakage}, a worst-case privacy criterion
that enforces two simultaneous constraints for every disclosed symbol $u\in\mathcal{U}$: (i) $u$ may be correlated
with at most $N$ realizations of $X$, and (ii) the total leakage toward those realizations is bounded. This notion
captures scenarios in which each disclosure should affect only a small subset of sensitive outcomes while maintaining
strict control over per-output leakage.
In the high-privacy regime, we use concepts from information geometry to obtain a local quadratic approximation of
mutual information which measures utility between $U$ and $Y$. When the leakage matrix $P_{X|Y}$ is invertible, this approximation reduces the design problem to
a sparse quadratic maximization, known as the Rayleigh-quotient problem, with an $\ell_0$ constraint. We further show that,
for the approximated problem, one can without loss of optimality restrict attention to a binary released variable
$U$ with a uniform distribution. For small alphabet sizes, the exact sparsity-constrained optimum can be computed via
combinatorial support enumeration, which quickly becomes intractable as the dimension grows.
For general dimensions, the resulting sparse Rayleigh-quotient maximization is NP-hard and closely related to sparse
principal component analysis (PCA). We propose a convex semidefinite programming (SDP) relaxation that is solvable in
polynomial time and provides a tractable surrogate for the NP-hard design, together with a simple rounding procedure
to recover a feasible leakage direction. We also identify a sparsity threshold beyond which the sparse optimum
saturates at the unconstrained spectral value and the SDP relaxation becomes tight. Numerical experiments demonstrate
a sharp phase-transition behavior and show that beyond this threshold the SDP solution matches the exact
combinatorial optimum.

\end{abstract}


\section{Introduction}
As shown in Fig.~\ref{sys1}, an agent wants to disclose some useful data denoted by a random variable (RV) $Y$ to a user. Here, $Y$ is arbitrarily correlated with the private data denoted by RV $X$. Furthermore, the agent uses a privacy mechanism to produce the disclosed data denoted by RV $U$. The agent’s goal is to design $U$ based on $Y$ which discloses as much information as possible about $Y$ while satisfying a privacy criterion. 
We use mutual information to measure utility, and we introduce a new metric, called \emph{sparse point-wise privacy leakage}, to quantify privacy. Sparse point-wise privacy leakage consists of two constraints. The first constraint bounds the number of realizations of $X$ that can be correlated with $u \in \mathcal{U}$. The second constraint bounds the amount of leakage from $X$ to $u$ through the correlated letters of $X$. For the first constraint, we use the $\ell_0$-norm and for the second, we use the $\chi^2$-distance.

Sparse point-wise privacy leakage is a practical measure, as in many scenarios 
$u\in\mathcal{U}$ cannot disclose information about some 
$x\in\mathcal{X}$ and can leak information about only a few of them. For instance, in medical data sharing, the sensitive variable $X$ may represent a patient’s exact diagnosis, while the disclosed variable $U$ corresponds to a medical summary released to a third party; it is undesirable for a single summary $u$ to be informative about many diagnoses. Sparse point-wise privacy leakage ensures that each $u$ is correlated with only a small subset of medical conditions, thereby limiting worst-case inference while preserving utility.


\begin{figure}[]
	\centering
	\includegraphics[width = 0.5\textwidth]{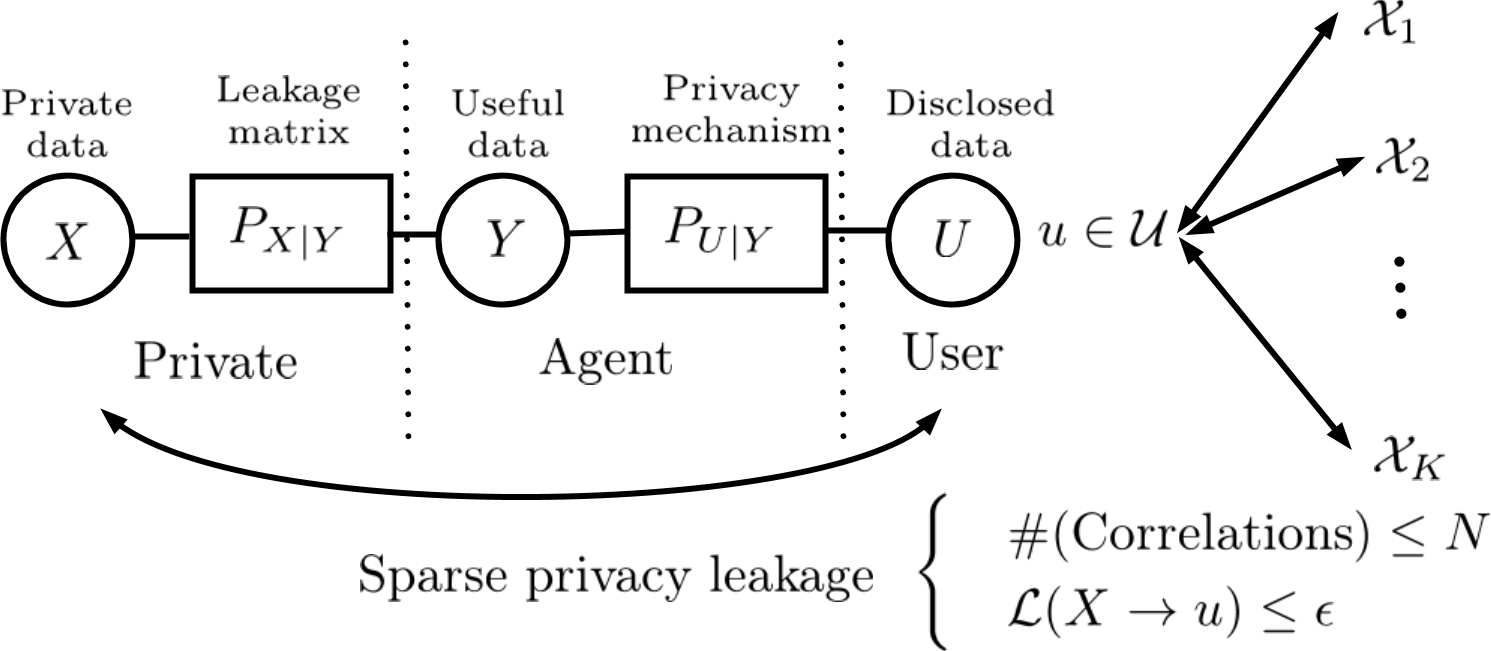}
	\caption{In this model, disclosed data $U$ is designed which maximizes the information released about $Y$ and satisfies the sparse point-wise privacy leakage constraint. Here, each realization of $U$ must satisfy two constraints:
1. each $u$ can be correlated with at most $N$ letters of $X$; and
2. the total leakage from $X$ to $u$ is bounded.
Here, $\mathcal{L}(X \rightarrow u_i)$ denotes the amount of leakage from $X$ (through the correlated letters) to $u \in \mathcal{U}$.
}
	\label{sys1}
\end{figure}
\subsection{Related works}
Related works on the information-theoretic privacy design can be found in \cite{borz,koala,shah,khodam,Khodam22,zarab1,zarab2,makhdoumi, Total, Calmon1,oof,razegh,emma,ang,sep,lopuha,long}. 
In \cite{koala}, \emph{secrecy by design} problem is studied under the perfect secrecy assumption. Bounds on secure decomposition have been derived using the Functional Representation Lemma. 
In \cite{shah}, the privacy problems considered in \cite{koala} are generalized by relaxing the perfect secrecy constraint and allowing some leakage. 
In \cite{borz}, the problem of privacy-utility trade-off considering mutual information both as measures of privacy and utility is studied. Under perfect privacy assumption, it has been shown that the privacy mechanism design problem can be reduced to linear programming. 
The concept of \emph{lift} is studied in \cite{zarab2} and represents the likelihood ratio between the posterior and prior beliefs concerning sensitive features within a dataset. In \cite{Calmon2}, fundamental limits of the privacy utility trade-off measuring the leakage using estimation-theoretic guarantees are studied.

In \cite{khodam} and \cite{Khodam22}, two point-wise privacy leakage have been introduced.
As discussed in \cite{khodam} and \cite{Khodam22}, it may not be desirable to use average measures to quantify privacy leakage, since some data points (realizations) may leak more information than a prescribed threshold. In other words, if an adversary has access to such letters, it can infer a considerable amount of information about the sensitive data $X$. 
 It has been shown that by using concepts from information geometry, the main complex design problem can be approximated by linear algebra problems.
This method has been used in various network information theory problems, as well as in privacy, secrecy, and fairness design problems \cite{Shashi,huang,khodam,Khodam22,shah,razegh,emma,ang,long,zamani2025fair,zamani2025fair2}. 
Specifically, this approach has been applied to point-to-point and selected broadcast channels \cite{Shashi,huang}, as well as to privacy mechanism design \cite{khodam,Khodam22,shah,razegh,emma,long,ang}. In particular, \cite{ang} approximates mutual information and relative entropy by quadratic functions, while \cite{long} considers Local Information Privacy (LIP) and max-lift as privacy leakage measures. The same framework has also been used to design fair mechanisms under bounded demographic parity and equalized odds constraints \cite{zamani2025fair,zamani2025fair2}. Additionally, information-geometric methods have enabled local approximations of secrecy capacity for wiretap channels \cite{emma} and hypothesis testing problems under bounded mutual information constraints in high-privacy regimes \cite{ang}.

\subsection{Contributions}
This paper develops a privacy mechanism design framework under a new \emph{sparse point-wise leakage} requirement:
each disclosed symbol $u$ is allowed to be informative about only a limited number of sensitive letters of $X$, and
even for those letters the leakage must remain small. Our main contributions are:

\begin{enumerate}
\item \textbf{A sparse point-wise privacy criterion.}
We introduce \emph{sparse point-wise privacy leakage}, which imposes two per-output constraints: (i) each disclosed
symbol $u$ can be correlated with at most $N$ realizations of $X$ (a sparsity constraint), and (ii) the total
per-output leakage toward those realizations is bounded (via a point-wise $\chi^2$ constraint). This criterion
captures operational scenarios where disclosures should be \emph{localized} to a small subset of sensitive outcomes
while controlling worst-case leakage.

\item \textbf{High-privacy approximation via information geometry.}
In the high-privacy regime (small leakage budget), we apply the information-geometric local approximation framework
of~\cite{khodam,Khodam22} to obtain a quadratic approximation of the privacy--utility tradeoff. Under an invertible
leakage matrix $P_{X|Y}$, the resulting mechanism design reduces to a \emph{sparse quadratic maximization} problem
equivalently expressed as a sparse Rayleigh-quotient over the subspace orthogonal to $\sqrt{P_X}$.

\item \textbf{Structural simplification: binary outputs suffice (for the approximation).}
For the approximated problem, we prove that one can without loss of optimality restrict attention to a \emph{binary}
released variable $U$ with a uniform distribution. This converts the design into selecting a single sparse
\emph{leakage direction}, greatly simplifying both analysis and computation.

\item \textbf{Exact characterization for small dimensions and the NP-hardness barrier.}
For small alphabet sizes, the sparse Rayleigh-quotient optimum (and the corresponding Pareto value as a function of
the sparsity budget $N$) can be computed exactly by combinatorial support enumeration. However, for general
dimensions this sparse quadratic maximization is NP-hard (closely related to sparse PCA), motivating tractable
relaxations.

\item \textbf{A polynomial-time convex SDP relaxation and recovery.}
We propose a convex semidefinite programming (SDP) relaxation that lifts the sparse quadratic maximization to the
matrix domain, drops the rank-one constraint, and promotes sparsity via an entrywise $\ell_1$ budget. The SDP is
solvable in polynomial time and provides an efficiently computable upper bound; combined with a simple rounding step,
it yields a feasible sparse leakage direction.

\item \textbf{Tightness (exactness) threshold and empirical transition.}
We characterize a deterministic saturation/exactness threshold: once the sparsity budget is large enough to contain the
dominant spectral leakage direction (restricted to $\sqrt{P_X}^{\perp}$), the sparse optimum saturates at the
unconstrained spectral value and the SDP relaxation becomes tight (rank-one). Numerical experiments exhibit a sharp
phase transition and show that beyond this threshold the rounded SDP solution matches the exact combinatorial
optimum.
\end{enumerate}

\noindent\textbf{Paper organization.}
Section~II introduces the sparse point-wise leakage criterion and the privacy design problem.
Section~III derives the local approximation and shows that binary $U$ is sufficient, reducing the design to a sparse
Rayleigh quotient. Section~IV connects this optimization to sparse PCA, proposes an SDP relaxation, and establishes a
deterministic tightness threshold. Section~V provides numerical experiments illustrating a sharp phase transition.
\section{System model and Problem Formulation} \label{system}
Let $P_{XY}$ denote the joint distribution of discrete random variables $X$ and $Y$ defined on finite alphabets $\cal{X}$ and $\cal{Y}$ with equal cardinality, i.e, $|\cal{X}|=|\cal{Y}|=\mathcal{K}$.
We represent $P_{XY}$ by a matrix defined on $\mathbb{R}^{|\mathcal{K}|\times|\mathcal{K}|}$ and 
marginal distributions of $X$ and $Y$ by vectors $P_X$ and $P_Y$ defined on $\mathbb{R}^{|\mathcal{K}|}$ and $\mathbb{R}^{|\mathcal{K}|}$ given by the row and column sums of $P_{XY}$. 
We assume that each element in vectors $P_X$ and $P_Y$ is non-zero. Furthermore, 
we represent the leakage matrix $P_{X|Y}$ by a matrix defined on $\mathbb{R}^{|\mathcal{K}|\times|\mathcal{K}|}$, which is assumed to be invertible. Furthermore, for given $u\in \mathcal{U}$, $P_{X,U}(\cdot,u)$ and $P_{X|U}(\cdot|u)$ defined on $\mathbb{R}^{|\mathcal{X}|}$ are distribution vectors with elements $P_{X,U}(x,u)$ and $P_{X|U}(x|u)$ for all $x\in\cal X$. 
The relation between $U$ and $Y$ is described by the kernel $P_{U|Y}$ defined on $\mathbb{R}^{|\mathcal{U}|\times|\mathcal{Y}|}$. 
In this work, $P_{X}(x)$, $P_{X}$, $\sqrt{P_{X}}$ and $[P_{X}]$ denote $P_{X}(X=x)$, distribution vector of $X$, a vector with entries $\sqrt{P_X(x)}$, and a diagonal matrix with diagonal entries equal to $P_{X}(x)$, respectively.
Here, $\ell_2$ and $\ell_0$-norm of a vector $V$ are denoted by $\|V\|_2$ and $\|V\|_0$, respectively. Furthermore, we define $W\triangleq [\sqrt{P_Y}^{-1}]P_{X|Y}^{-1}[\sqrt {P_X}]$. In this work, $\sigma^2_1\leq \sigma^2_2\ldots \leq\sigma^2_{\mathcal{K}} $ and $V_{1}, \ldots,V_{\mathcal{K}}$ correspond to the singular values and vectors of $W$, respectively, where $\|V_i\|_{2}=1$. 
Furthermore, we assume that the Markov chain $X - Y - U$ holds.
Our goal is to design the privacy mechanism that produces the disclosed data $U$, which maximizes the utility and satisfies the sparse point-wise privacy criterion.
In this work, utility is measured by the mutual information $I(U;Y)$, while privacy leakage is quantified by the following point-wise constraints: 
\begin{subequations}\label{sp}
\begin{align}
    \|P_{X|U=u}-P_X\|_{0}&\leq N, \ \forall u\in\mathcal{U},\label{l0}\\
    \mathcal{L}(X \rightarrow u)\triangleq \chi^2(P_{X|U=u};P_X)&\leq 
    \epsilon^2, \ \forall u\in\cal U,\label{l2}
    \end{align}
\end{subequations}
where the $\ell_0$-norm corresponds to the number of non-zero element of a vector.
 Thus, the privacy problem can be stated as follows 
\begin{align}\label{problem}
    g_{\epsilon}^N(P_{XY})&\triangleq\sup_{\begin{array}{c} 
	\substack{P_{U|Y}: X-Y-U,\\ \|P_{X|U=u}-P_X\|_{0}\leq N, \forall u\in\mathcal{U}, \\ \chi^2(P_{X|u};P_X)\leq\epsilon^2,\ \forall u\in\mathcal{U},}
	\end{array}}I(Y;U),
\end{align}
In this paper, \eqref{l0} specifies the maximum number of letters (alphabets) in $\mathcal{X}$ that can be correlated with each $u \in \mathcal{U}$. Furthermore, \eqref{l2} limits the strength of the correlations between such letters and $u \in \mathcal{U}$.
\begin{remark}
\normalfont
We refer to \eqref{sp} as the \emph{sparse point-wise privacy leakage} constraint, since each $u$ must satisfy both \eqref{l0} and \eqref{l2}. Furthermore, for small $N$, the vector $P_{X|U=u}-P_X$, which reflects the correlation between the letters of $X$ and letter $u$, becomes sparse. 
\end{remark}
\begin{remark}
\normalfont
By letting $N=|\mathcal{X}|=\mathcal{K}$, \eqref{problem} leads to the problem studied in \cite{khodam}. Furthermore, setting either $\epsilon = 0$ or $N=0$, reduces the model to the perfect privacy problem studied in \cite{borz}.
\end{remark}
\begin{remark}
\normalfont
Let $P_{X|U} - [P_X,\ldots,P_X]$ be a matrix whose columns are $P_{X|U=u} - P_X$. Then, \eqref{l0} implies that each column has at most $N$ nonzero elements, and \eqref{l2} limits its energy.
\end{remark}
\begin{example} (Motivating example)
    Consider a medical data--sharing system in which the sensitive variable $X$ represents a patient’s exact diagnosis, drawn from a large set of possible diseases, while the disclosed variable $U$ corresponds to a medical summary released to a third party, such as a researcher or an insurance provider. If a single disclosed symbol $u$ is correlated with many diagnoses, then observing $u$ may allow an adversary to infer significant information about the patient’s underlying condition. To reduce this risk, it is desirable to enforce a sparsity constraint whereby each $u \in \mathcal{U}$ is correlated with only a limited number of diagnoses. For example, a summary indicating cardiovascular risk may be linked to only a few related conditions rather than a broad range of diseases, thereby limiting worst-case disclosure while preserving utility for analysis or decision-making.
\end{example}
\begin{example} (Motivating example $2$)
    Consider a location-based service in which the sensitive variable $X$ denotes a user’s exact location, while the disclosed variable $U$ corresponds to a coarse location label such as \emph{home} or \emph{work}. If a single disclosed symbol $u$ is correlated with many possible locations, observing $u$ can significantly narrow down the user’s true location. A sparse leakage constraint mitigates this risk by ensuring that each $u \in \mathcal{U}$ is correlated with only a small subset of locations, thereby limiting worst-case inference while preserving utility.
\end{example}
\begin{remark}
    \normalfont
    As motivated above, we intend to keep $N$ small; however, we do not want $N$ to be too small. For instance, if $N = 1$, that is, each $u \in \mathcal{U}$ is correlated with only one realization of $X$, guessing may be easier than when $N$ is larger. Thus, $N$ should be small but not too small. We show later that $N = 1$ is not feasible, as it leads to zero utility; see Lemma \ref{lem1} and Remark~\ref{mohem}.
\end{remark}
\begin{lemma}\label{lem1}
    For any $\epsilon\geq 0$, $g_{\epsilon}^{N=1}(P_{XY})=0$. Hence, to attain non-zero utility we must have $N\geq 2$.
\end{lemma}
\begin{proof}
For $N = 1$, only one element of the vector $P_{X|U=u} - P_X$ can be nonzero. However, since
\[
\sum_x \bigl(P_{X|U=u}(x) - P_X(x)\bigr) = 0,
\]
we must have $P_{X|U=u}(x) = P_X(x)$ for all $x$ and $u$. In other words, all elements of $P_{X|U=u} - P_X$ must be zero.

This implies a perfect privacy constraint. Moreover, since $P_{X|Y}$ is invertible, using the Markov chain $X - Y - U$ and multiplying both sides of $P_{X|U=u}(x) = P_X(x)$ by $P_{X|Y}^{-1}$, we obtain
\[
P_{Y|U=u}(y) = P_Y(y), \quad \forall\, u, y,
\]
which implies that $U$ and $Y$ are independent, i.e.,
\[
I(Y;U) = 0.
\]
The final conclusion can also be obtained from \cite[Remark~1]{Khodam22}.

\end{proof}
\begin{proposition}\label{car1}
	It suffices to consider $U$ such that $|\mathcal{U}|\leq|\mathcal{Y}|$. Furthermore, a maximum can be used in \eqref{problem} since the
		corresponding supremum is achieved.
\end{proposition}
\begin{proof}
    The proof follows arguments similar to those in \cite[Proposition~1]{khodam}.
\end{proof}

\section{Main results}
In this section, we first approximate \eqref{problem}. Before stating the next result, we rewrite the conditional distribution $P_{X|U=u}$ as a perturbation of $P_X$. Thus, for any $u\in\mathcal{U}$, we can write $P_{X|U=u}=P_X+\epsilon\cdot J_u$, where $J_u\in\mathbb{R}^{|\mathcal{X}|}$ is a perturbation vector which satisfies following properties.
\begin{align}
\sum_{x\in\mathcal{X}} J_u(x)&=0,\ \forall u,\label{proper1}\\
\sum_{u\in\mathcal{U}} P_U(u)J_u(x)&=0,\ \forall x\label{proper2},\\
\|[\sqrt{P_X}^{-1}]J_u\|_{2}^2&= \sum_{x\in\mathcal{X}} \frac{J_u(x)^2}{P_X(x)}\leq 1, \ \forall u, \label{proper3}\\
\|J_u\|_{0}&\leq N, \ \forall u.\label{proper4}
\end{align}
The first two properties ensure that $P_{X|U=u}$ is a valid probability distribution \cite{khodam,Khodam22}, and the third and fourth properties follows from \eqref{l0} and \eqref{l2}, respectively.
Furthermore, 
letting $L_u\triangleq[\sqrt{P_X}^{-1}]J_u$, we can rewrite the constraints as follows
\begin{align}
L_u\perp \sqrt{P_X},\ &\forall u,\label{proper11}\\
\sum_{u} P_U(u)L_u&=0,\ \label{proper22},\\
\|L_u\|_{2}&\leq 1, \ \forall u, \label{proper33}\\
\|L_u\|_{0}&\leq N, \ \forall u.\label{proper44}
\end{align}
In the next result, we approximate \eqref{problem} by a quadratic function. 
To do so, we use the mutual information approximation stated in \cite[Proposition 3]{khodam}.

\begin{proposition}\label{prop1}
    For sufficiently small $\epsilon$, \eqref{problem} can be approximated by the following problem
    \begin{align}\label{tt}
        \max_{\begin{array}{c} 
		\substack{L_u,P_U: \sum_i P_U(u)L_{u}=0,\ \|L_{u}\|_{0}\leq N,\\ \|L_{u}\|_{2}\leq 1,\ L_{u}\perp \sqrt{P_X}, \forall u,}
		\end{array}} \!\!\!\!\!\!\!0.5\epsilon^2\left(\sum_u \!\!P_U(u_i)\|WL_{u}\|_{2}^2 \right),
    \end{align}
\end{proposition}
\begin{proof}
    The proof is similar to \cite{khodam} and is based on the second order Taylor approximation of the KL-divergence. The only difference is that, due to the sparsity condition, we have an extra constraint $\|L_u\|_{0}\leq N$.
\end{proof}
In the next result, we show that to solve \eqref{tt}, it suffices to consider a binary $U$ with uniform distribution.
\begin{proposition}\label{prop2}
    To solve \eqref{tt}, it suffices to consider a binary $U$ with uniform distribution and we have 
    \begin{align}\label{2}
    \eqref{tt}= \max_{\begin{array}{c} 
		\substack{L: \|L\|_{0}\leq N,\\ \|L\|_{2}\leq 1,\ L\perp \sqrt{P_X},}
		\end{array}}0.5\epsilon^2\|WL\|_{2}^2.
    \end{align}
\end{proposition}
\begin{proof}
   The proof is similar to that of \cite[Proposition~4]{khodam}. Let $L^*$ denote the maximizer of the right-hand side of \eqref{prop2}. In this case, we set $P_U(u_1)=P_U(u_2)=\tfrac{1}{2}$ and $L_{u_1}=-L_{u_2}=L^*$. Clearly, $L^*$ and $-L^*$, with equal weights $P_U(u_1)=P_U(u_2)=\tfrac{1}{2}$, are feasible in \eqref{tt}. Using arguments similar to those in \cite[Proposition~4]{khodam}, we obtain the result.
\end{proof}
\begin{remark}
\normalfont
Using Proposition \ref{prop2}, it suffices to solve the right hand side in \eqref{2}. 
\end{remark}
\begin{corollary}
    We have
    \begin{align}
        \eqref{2}\leq \frac{1}{2}\epsilon^2\sigma_{\mathcal{K}}^2.
    \end{align}
    Furthermore, if $L=V_{\mathcal{K}}$ satisfies $\|L\|_{0}\leq N$, we have 
    \begin{align}
        \eqref{2}=\frac{1}{2}\epsilon^2\sigma_{\mathcal{K}}^2.
    \end{align}
\end{corollary}
\begin{proof}
    The upper bound is obtained by removing the constraint $\|L\|_{0}\leq N$. Furthermore, if $V_{\mathcal{K}}$ satisfies $\|V_{\mathcal{K}}\|_{0}\leq N$, the upper bound is attained.
\end{proof}
Next, we rewrite the optimization problem on the right-hand side of \eqref{2}. Before stating the next result, we recall that, using \cite[Appendix C]{khodam}, the smallest singular value of $W$ is $1$ with corresponding singular vector $\sqrt{P_X}$, i.e., $\sigma_1=1$ and $V_1=\sqrt{P_X}$. 
\begin{proposition}
    The maximization problem in \eqref{2} can be reformulated as follows
        \begin{align}\label{22}
    \eqref{2}= \max_{\begin{array}{c} 
		\substack{\{\alpha_i\}: \sum_{i=2}^{\mathcal{K}}\alpha_i^2\leq 1,\\ \|\sum_{i=2}^{\mathcal{K}}\alpha_iV_i\|_{0}\leq N,}
		\end{array}}0.5\epsilon^2\sum_{i=2}^{\mathcal{K}}\alpha_i^2\sigma_i^2.
    \end{align}
\end{proposition}
\begin{proof}
    Using $V_1=\sqrt{P_X}$, we have
    \begin{align*}
        L\perp \sqrt{P_X}\leftrightarrow L\in\text{span}\{V_2,\ldots,V_{\mathcal{K}}\},
    \end{align*}
    since the singular vectors are mutually orthogonal.
    Hence, in \eqref{2}, we substitute $L$ by $\sum_{i=2}^{\mathcal{K}}\alpha_i V_i$. Then, we have
    \begin{align*}
        \|WL\|_{2}^2&=\sum_{i=2}^{\mathcal{K}}\alpha_i^2\sigma_i^2,\\
        \|L\|_{2}^2&= \sum_{i=2}^{\mathcal{K}}\alpha_i^2\leq 1,\\
        \|L\|_{0}&=\|\sum_{i=2}^{\mathcal{K}}\alpha_iV_i\|_{0}\leq N.
    \end{align*}
    This completes the proof.
\end{proof}
\subsection{From Sparse Leakage Design to Sparse PCA and an SDP Relaxation}\label{subsec:sdp}
\textbf{Connection to sparse Rayleigh-quotient maximization (sparse PCA):}
The local approximation reduces the mechanism design to selecting a \emph{leakage direction} $L$ that maximizes a
quadratic form under a cardinality constraint. Define
$A \triangleq W^{\mathsf T}W \succeq 0$, so that $\|WL\|_2^2 = L^{\mathsf T}AL$. Since the objective is homogeneous in
$L$, the norm constraint is tight at the optimum i.e. $\|L\|_{2}=1$ (whenever the optimum is nonzero), and the approximate design can be
written as a \emph{sparse Rayleigh-quotient} problem on the subspace $\sqrt{P_X}^{\perp}$:
\begin{align}
\max_{L}\quad & \frac{L^{\mathsf T}AL}{L^{\mathsf T}L}
\label{eq:sparse_rayleigh}
\\[-1mm]
\text{s.t.}\quad & L \perp \sqrt{P_X}, \qquad \|L\|_0 \le N .
\nonumber
\end{align}
Equivalently, \eqref{eq:sparse_rayleigh} maximizes the Rayleigh quotient of $A(S,S)$ over the orthogonal subspace $(\sqrt{P_X})^{\perp}_S$ for all  support set $S \subseteq [\mathcal{K}]$ with $|S|\le N$. This is the canonical primitive underlying \emph{sparse PCA}
(cardinality-constrained variance maximization), and is NP-hard in general due to the combinatorial search over
supports \cite{nphard,A.A.Amini,direct_pca,NesterovRichtarikSepulchre2010}. This connection motivates tractable surrogates such as semidefinite relaxations of
cardinality-constrained eigenvalue problems \cite{direct_pca} and efficient heuristic
methods (e.g., generalized power iterations) developed for sparse PCA \cite{NesterovRichtarikSepulchre2010}.

\begin{remark}\label{mohem}
\normalfont
The orthogonality constraint $L \perp \sqrt{P_X}$ implies that the smallest feasible sparsity level is at least two
under our full-support assumption $P_X(x)>0$ for all $x\in\mathcal{X}$. Indeed, any $1$-sparse vector
$L=c e_i$ satisfies $\langle L,\sqrt{P_X}\rangle=c\sqrt{P_X(i)}\neq 0$ unless $c=0$, hence it cannot satisfy
$L\perp \sqrt{P_X}$ except trivially. This justifies Lemma \ref{lem1}. Therefore, nontrivial utility requires $N\ge 2$. Moreover, in case of $N=2$, there exists a $2$-sparse non-zero $L$ supported on ${S}=\{i,j\}$ for any $i\neq j$ satisfying $L\perp {P_X}({S})$ (choose for instance $L_i=P_X(j), L_j=-P_X(i)$).
\end{remark}
For fixed $N\ge 2$ and moderate alphabet size $\mathcal{K}$, one can compute the exact optimum of
\eqref{eq:sparse_rayleigh} by enumerating all supports $S\subseteq[\mathcal{K}]$ with $|S|\le N$ and solving the
resulting Rayleigh-quotient problem restricted to the coordinates in $S$ (with the linear constraint
$\langle L,\sqrt{P_X}\rangle=0$ imposed on that support). However, the number of candidate supports grows as
$\binom{\mathcal{K}}{N}$, and the problem quickly becomes computationally intractable as $\mathcal{K}$ increases.
This motivates polynomial-time surrogates, in particular convex relaxations inspired by the sparse PCA literature
\cite{direct_pca,NesterovRichtarikSepulchre2010,A.A.Amini}.

Proposition~\ref{prop2} reduces the local privacy-utility design to the following quadratic maximization:
\begin{align}\label{eq:main_noncvx_L}
U_{\rm OPT}(N)\triangleq\max_{L\in\mathbb{R}^{\mathcal{K}}}\quad & \frac{1}{2}\epsilon^2\|WL\|_{2}^2 \\
\text{s.t.}\quad & \|L\|_{2}\le 1,\ \ L\perp \sqrt{P_X},\ \ \|L\|_{0}\le N.\nonumber
\end{align}
Problem \eqref{eq:main_noncvx_L} is equivalent to \eqref{eq:sparse_rayleigh} whenever the optimum is nonzero.
We call $U_{\rm OPT}(N)$ the optimum value function (privacy-utility trade-off) because it gives the maximum achievable utility under the sparsity (leakage budget ) constraint $N$.  Varying $N$ traces the Pareto tradeoff between the leakage and utility.
We refer to $L$ as the \emph{leakage direction} since $P_{X|U=u}-P_X=\epsilon[\sqrt{P_X}]\,L$.
The constraint $\|L\|_{0}\le N$ enforces \emph{sparse point-wise leakage}: each disclosed symbol $u$
perturbs at most $N$ letters of $X$.
As discussed above, the $\ell_0$-constrained Rayleigh quotient is NP-hard in general and closely related to sparse PCA.

\textbf{A lifted rank-one view (low-rank + sparse)}:
A standard technique is to lift $L$ to a positive semidefinite matrix $X$:
\begin{align}\label{eq:lift_def}
X \triangleq LL^T \in \mathbb{S}_+^{\mathcal{K}}.
\end{align}
Then $\mathrm{tr}(X)=\|L\|_{2}^2\le 1$, $X\sqrt{P_X}=0$ (since $L\perp \sqrt{P_X}$), and
\begin{align}
L^TAL=\mathrm{tr}(AL L^T)=\mathrm{tr}(AX).
\end{align}
Moreover, the \emph{low-rank} structure is explicit: $X$ must satisfy $\mathrm{rank}(X)=1$.
The sparsity of $L$ induces sparsity in $X$ as well (indeed, for rank-one $X=LL^T$, the support of $X$
is essentially the Cartesian product of the support of $L$).
Therefore, \eqref{eq:main_noncvx_L} is equivalently a \emph{rank-one sparse semidefinite program}:
\begin{align}\label{eq:noncvx_lifted}
&U_{\rm OPT}=\max_{X\succeq 0}\quad  \frac{1}{2}\epsilon^2\,\mathrm{tr}(AX) \nonumber\\
&\text{s.t.}\quad \mathrm{tr}(X)\le 1,\ \ X\sqrt{P_X}=0,\ \ \mathrm{rank}(X)=1, \|X\|_{0}\le N^2 
\end{align}
The nonconvexity is now concentrated into two structural requirements:
\emph{rank-one} and \emph{sparsity}. This is precisely the ``low-rank + sparse'' regime.

\paragraph{Convex SDP relaxation }
To obtain a tractable design with provable properties, we relax both nonconvex structures in the standard way:
(i) drop the rank-one constraint, and (ii) replace sparsity by the entrywise $\ell_1$-norm, the tightest convex
surrogate promoting sparsity of a matrix.
This yields the following convex SDP relaxation:
\begin{align}\label{eq:sdp_relax_final}
&{U}_{\rm SDP}(\tau)\triangleq
\max_{X\succeq 0} \frac{1}{2}\epsilon^2\,\mathrm{tr}(AX) \nonumber\\
&\text{s.t.}\ \mathrm{tr}(X)\le 1,\ \ X\sqrt{P_X}=0,\ \ \|X\|_{1,\mathrm{entry}}\le \tau,
\end{align}
where $\|X\|_{1,\mathrm{entry}}\triangleq \sum_{i,j}|X_{ij}|$ and $\tau>0$ controls the sparsity of the lifted
leakage matrix.
This relaxation is closely aligned with the semidefinite formulation of sparse PCA and is attractive because it
is convex, globally solvable, and produces an \emph{upper bound} on the original nonconvex optimum.
When the solution of \eqref{eq:sdp_relax_final} is rank-one, the relaxation is exact and directly yields the
optimal leakage direction.

\paragraph{Exactness beyond a leakage threshold}
The next theorem formalizes a deterministic threshold phenomenon: once the sparsity budget is large enough to
accommodate the \emph{unconstrained} optimal leakage direction, the nonconvex problem saturates at the global
spectral upper bound and the convex SDP becomes tight (rank-one). This explains the empirically observed
``transition'' where the SDP solution matches the Pareto-optimal nonconvex design.
\begin{theorem}[Saturation of the sparse Rayleigh-quotient and tightness/uniqueness of the SDP]\label{thm:full_tightness}
Let $p \triangleq \sqrt{P_X}\in\mathbb{R}^{\mathcal{K}}$ (note that $\|p\|_2^2=\sum_x P_X(x)=1$), and let
$A\triangleq W^{\mathsf T}W\succeq 0$.
Define the top Rayleigh-quotient value on the subspace $p^\perp$:
\begin{align}
\lambda_\star \triangleq \max_{\|v\|_2=1,\ v\perp p} v^{\mathsf T}A v,
\qquad
v_\star \in \mathop{\arg\max}_{\|v\|_2=1,\ v\perp p} v^{\mathsf T}A v.
\end{align}
Define the thresholds
\begin{align}
N_{\mathrm{th}}\triangleq \|v_\star\|_0, \qquad
\tau_{\mathrm{th}}\triangleq \|v_\star v_\star^{\mathsf T}\|_{1,\mathrm{entry}}
= \|v_\star\|_1^2.
\end{align}
Consider the sparse Rayleigh-quotient problem
\begin{align}
U_{\mathrm{OPT}}(N)\triangleq
\max_{L}\quad & L^{\mathsf T}AL
\label{eq:l0_prob_thm}
\\
\text{s.t.}\quad & \|L\|_2=1,\ \ L\perp p,\ \ \|L\|_0\le N,
\nonumber
\end{align}
and the convex SDP relaxation
\begin{align}
U_{\mathrm{SDP}}(\tau)\triangleq
\max_{X\succeq 0}\quad & \langle A,X\rangle
\label{eq:sdp_prob_thm}
\\
\text{s.t.}\quad & \mathrm{tr}(X)\le 1,\ \ Xp=0,\ \ \|X\|_{1,\mathrm{entry}}\le \tau.
\nonumber
\end{align}
Then:
\begin{enumerate}
\item[\textup{(i)}] (\textup{Sparse saturation}) If $N\ge N_{\mathrm{th}}$, then
\(
U_{\mathrm{OPT}}(N)=\lambda_\star
\)
and $L^\star=\pm v_\star$ is optimal for \eqref{eq:l0_prob_thm}.
\item[\textup{(ii)}] (\textup{SDP tightness}) For any $\tau>0$, $U_{\mathrm{SDP}}(\tau)\le \lambda_\star$.
If $\tau\ge \tau_{\mathrm{th}}$, then $X^\star=v_\star v_\star^{\mathsf T}$ is feasible and achieves equality, hence
\(
U_{\mathrm{SDP}}(\tau)=\lambda_\star
\)
and the relaxation is tight (rank-one).
\item[\textup{(iii)}] (\textup{Uniqueness under a simple top eigenvalue}) If the maximizer on $p^\perp$ is unique up to sign
(equivalently, $\lambda_\star$ is a simple top eigenvalue of $A$ restricted to $p^\perp$),
then for any $\tau\ge \tau_{\mathrm{th}}$ the SDP optimizer is unique and equals
$X^\star=v_\star v_\star^{\mathsf T}$, and the sparse optimizer in \eqref{eq:l0_prob_thm} is unique up to sign.
\end{enumerate}
\end{theorem}

\begin{proof}
\textbf{(i)} Since $\|v_\star\|_0=N_{\mathrm{th}}\le N$, the vector $v_\star$ is feasible for \eqref{eq:l0_prob_thm}
whenever $N\ge N_{\mathrm{th}}$. Hence $U_{\mathrm{OPT}}(N)\ge v_\star^{\mathsf T}A v_\star=\lambda_\star$.
On the other hand, the feasible set of \eqref{eq:l0_prob_thm} is contained in
$\{v:\|v\|_2=1,\ v\perp p\}$, so for any feasible $L$ we have $L^{\mathsf T}AL\le \lambda_\star$ by definition of
$\lambda_\star$. Therefore $U_{\mathrm{OPT}}(N)\le \lambda_\star$. Combining both inequalities yields
$U_{\mathrm{OPT}}(N)=\lambda_\star$, and $L^\star=\pm v_\star$ is optimal.

\textbf{(ii)} Let $X$ be feasible in \eqref{eq:sdp_prob_thm}. Since $X\succeq 0$, write its eigen-decomposition
$X=\sum_{i}\mu_i q_i q_i^{\mathsf T}$ where $\mu_i\ge 0$, $\{q_i\}$ are orthonormal, and $\sum_i \mu_i=\mathrm{tr}(X)\le 1$.
The constraint $Xp=0$ implies that the range of $X$ is contained in $p^\perp$, hence any eigenvector with $\mu_i>0$
satisfies $q_i\perp p$. Therefore,
\begin{align}
\langle A,X\rangle
&= \sum_i \mu_i\, q_i^{\mathsf T}A q_i
\nonumber\\&\le \sum_i \mu_i\,\lambda_\star
\nonumber\\&\le \lambda_\star,
\end{align}
so $U_{\mathrm{SDP}}(\tau)\le \lambda_\star$ for all $\tau$.
Now assume $\tau\ge \tau_{\mathrm{th}}$. The rank-one matrix $X^\star=v_\star v_\star^{\mathsf T}$ satisfies
$X^\star\succeq 0$, $\mathrm{tr}(X^\star)=\|v_\star\|_2^2=1$, and $X^\star p=v_\star(v_\star^{\mathsf T}p)=0$.
Moreover,
\begin{align}
\|X^\star\|_{1,\mathrm{entry}}
&=\sum_{i,j}|(v_\star)_i(v_\star)_j|
\nonumber\\&=\Big(\sum_i |(v_\star)_i|\Big)^2
=
\|v_\star\|_1^2
\nonumber\\&=\tau_{\mathrm{th}}
\nonumber\\&\le \tau,
\end{align}
so it is feasible, and it achieves $\langle A,X^\star\rangle=v_\star^{\mathsf T}A v_\star=\lambda_\star$.
Hence $U_{\mathrm{SDP}}(\tau)=\lambda_\star$ and the SDP is tight (rank-one).

\textbf{(iii)} Assume $\lambda_\star$ is simple on $p^\perp$ (so the maximizer in $p^\perp$ is unique up to sign).
Let $\tau\ge \tau_{\mathrm{th}}$ and let $\widetilde X$ be any SDP optimizer. By (ii),
$\langle A,\widetilde X\rangle=\lambda_\star$ and necessarily $\mathrm{tr}(\widetilde X)=1$ (otherwise scaling up
increases the objective while preserving feasibility).
Write $\widetilde X=\sum_i \mu_i q_i q_i^{\mathsf T}$ as above with $\mu_i\ge 0$ and $\sum_i\mu_i=1$, and $q_i\perp p$
whenever $\mu_i>0$. Then
\begin{align}
\lambda_\star = \langle A,\widetilde X\rangle
&= \sum_i \mu_i\, q_i^{\mathsf T}A q_i
\nonumber\\&\le \sum_i \mu_i\,\lambda_\star
\nonumber\\&= \lambda_\star.
\end{align}
Therefore all inequalities must be equalities, which forces $q_i^{\mathsf T}A q_i=\lambda_\star$ for every $i$ with
$\mu_i>0$. By simplicity of the top eigenvalue on $p^\perp$, this implies $q_i=\pm v_\star$ for all $i$ with $\mu_i>0$.
Hence $\widetilde X=v_\star v_\star^{\mathsf T}$, proving uniqueness of the SDP optimizer.
Uniqueness (up to sign) of the sparse optimizer in \eqref{eq:l0_prob_thm} follows similarly: any maximizer must achieve
$\lambda_\star$ on $p^\perp$, hence must equal $\pm v_\star$.
\end{proof}

\begin{remark}\normalfont
Theorem~\ref{thm:full_tightness} formalizes a deterministic \emph{saturation} phenomenon for the $\ell_0$-constrained
Rayleigh-quotient: once the sparsity budget satisfies $N\ge N_{\mathrm{th}}$, the sparsity constraint becomes inactive
and the optimum attains the unconstrained value $\lambda_\star$ on $p^\perp$.
Similarly, for $\tau\ge \tau_{\mathrm{th}}$ the semidefinite relaxation is no longer conservative and admits the same
rank-one optimizer $X^\star=v_\star v_\star^{\mathsf T}$, explaining the regime in which the relaxed design matches the
optimal sparse value curve observed in simulations.
\end{remark}

\section{Numerical Experiments}\label{sec:experiments}
\begin{figure}[t]
\centering
\begin{subfigure}[t]{0.48\linewidth}
    \centering
    \includegraphics[width=\linewidth]{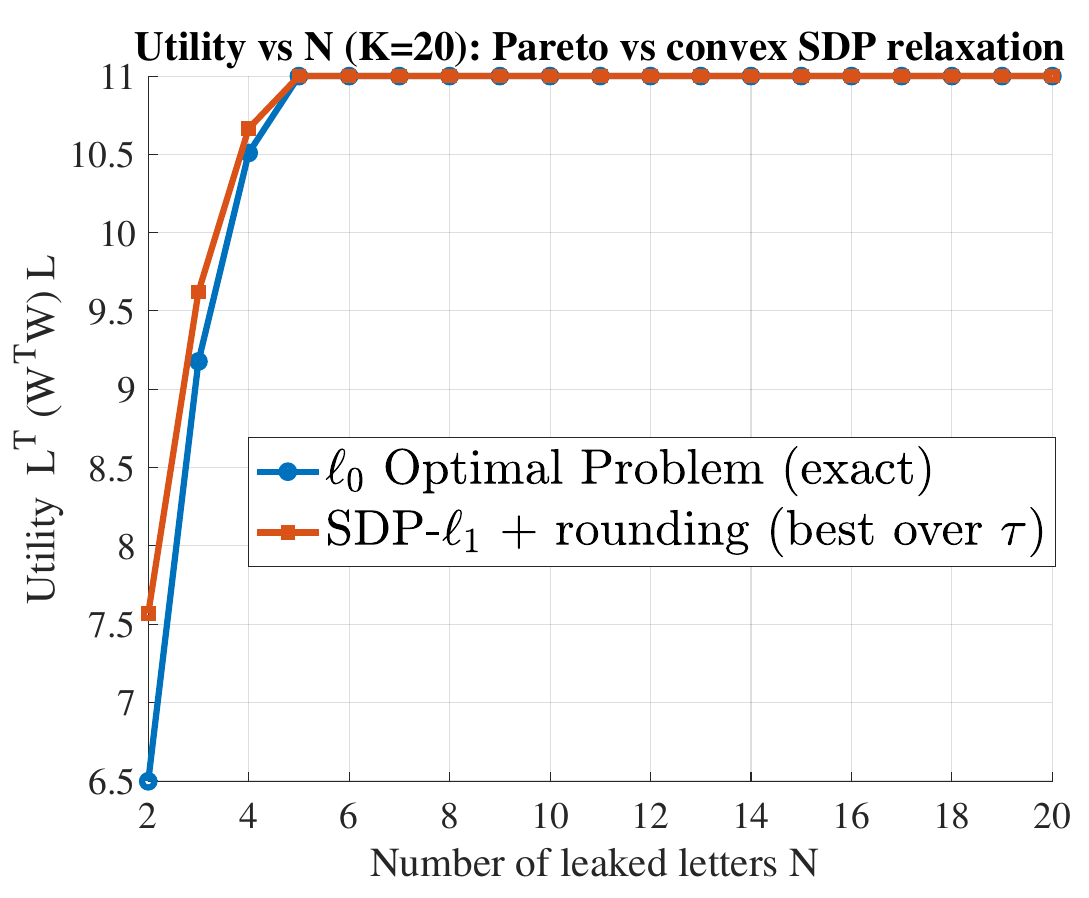}
    \caption{Utility versus leakage budget $N$.
    Exact Pareto optimum $U_{\mathrm{OPT}}(N)$ and rounded SDP envelope
    $\widehat{U}_{\mathrm{SDP}}(N)=\max_{\tau}\widehat{U}_{\mathrm{SDP}}(N,\tau)$.}
    \label{fig:utilityN_short}
\end{subfigure}
\hfill
\begin{subfigure}[t]{0.48\linewidth}
    \centering
    \includegraphics[width=\linewidth]{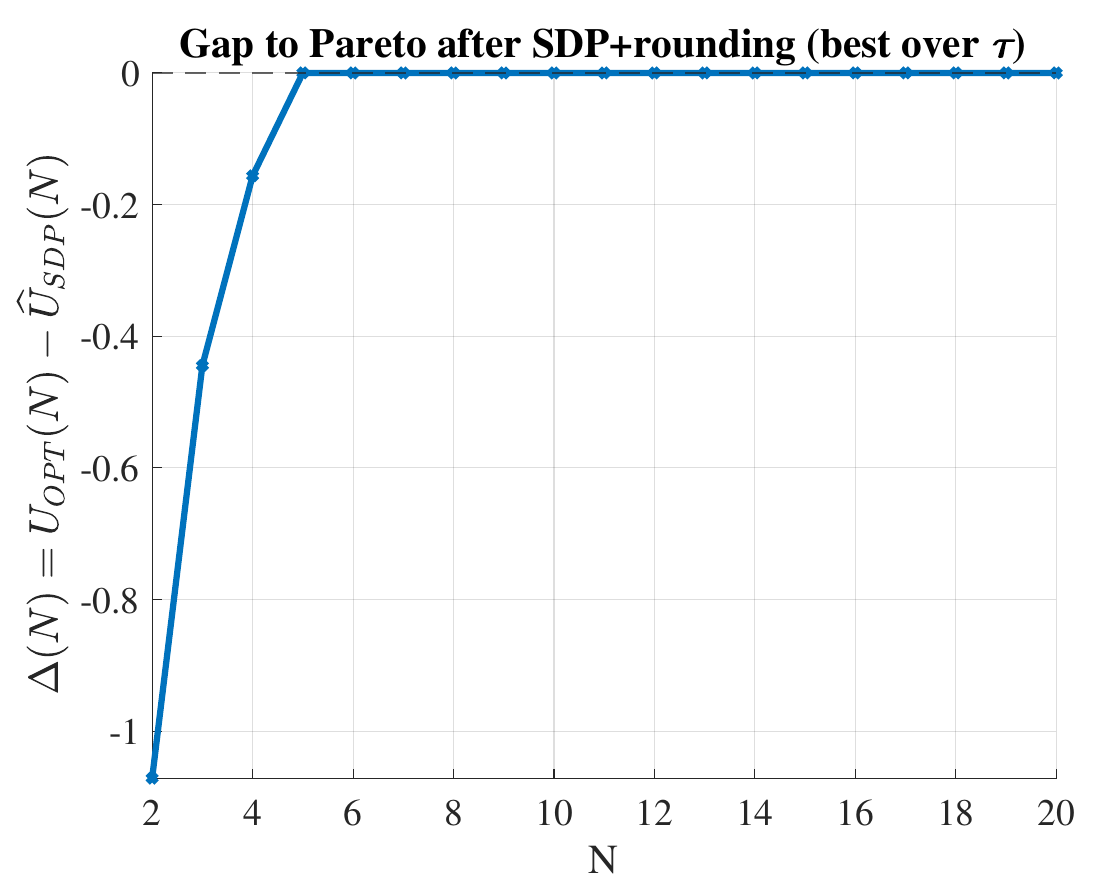}
    \caption{Pareto gap $\Delta(N)=U_{\mathrm{OPT}}(N)-\widehat{U}_{\mathrm{SDP}}(N)$.
    The gap collapses to numerical zero beyond $N_{\mathrm{th}}^{(\mathrm{emp})}$.}
    \label{fig:gapN_short}
\end{subfigure}
\caption{Tightness of the SDP relaxation as a function of the leakage budget $N$.
Beyond a sparsity threshold $N_{\mathrm{th}}^{(\mathrm{emp})}$, the convex SDP relaxation
(with simple rounding) matches the exact sparse Pareto optimum.}
\label{fig:tightness_summary}
\end{figure}
We evaluate the proposed SDP relaxation and rounding procedure for the sparse leakage-direction design of
Proposition~\ref{prop2}. For a fixed instance $(W,P_X)$, the utility of a leakage direction
$L\in\mathbb{R}^{\mathcal{K}}$ is
\begin{equation}\label{eq:expU_clean_nosubsec}
U(L)\triangleq\tfrac{1}{2}\epsilon^2\|WL\|_2^2=\tfrac{1}{2}\epsilon^2\,L^TAL,\qquad A=W^TW\succeq 0,
\end{equation}
under the constraints $\|L\|_2\le1$, $L\perp \sqrt{P_X}$, and $\|L\|_0\le N$, where $N$ is the leakage budget (number of
leaked letters). For moderate $\mathcal{K}$ we compute the exact optimal benchmark
\[
U_{\mathrm{OPT}}(N)\triangleq\max_{\|L\|_2\le 1,\ L\perp \sqrt{P_X},\ \|L\|_0\le N}\ \tfrac{1}{2}\epsilon^2\,L^TAL,
\]
by enumerating all supports of size $N$ and solving the resulting constrained Rayleigh-quotient problem on each
support.

We compare this ground truth to the convex lifted relaxation in the variable $X\approx LL^T$ shown by ${U}_{\rm SDP}(\tau)$ in \eqref{eq:sdp_relax_final}.
We sweep $\tau$
over a grid; for each $\tau$ we extract the top eigenvector of the SDP optimizer $X^\star(\tau)$, project it onto
$\sqrt{P_X}^{\perp}$, and enforce sparsity by hard-thresholding to the $N$ largest-magnitude entries (followed by
re-projection and normalization). This produces a rounded direction $\widehat{L}(N,\tau)$ and the rounded utility
$\widehat{U}_{\mathrm{SDP}}(N,\tau)=\tfrac{1}{2}\epsilon^2\,\widehat{L}(N,\tau)^TA\widehat{L}(N,\tau)$.
We report the best-over-$\tau$ envelope
\[
\widehat{U}_{\mathrm{SDP}}(N)\triangleq \max_{\tau}\widehat{U}_{\mathrm{SDP}}(\tau),
\]
and the Pareto gap $\Delta(N)\triangleq U_{\mathrm{OPT}}(N)-\widehat{U}_{\mathrm{SDP}}(N)$.

Figure ~\ref{fig:utilityN_short} plots the exact Pareto curve $U_{\mathrm{P}}(N)$ and the rounded SDP envelope
$\widehat{U}_{\mathrm{SDP}}(N)$ versus the leakage budget $N$. For small $N$, the rounded SDP is below the exact Pareto
value; beyond an empirical threshold $N_{\mathrm{th}}^{(\mathrm{emp})}$ the curves coincide to numerical precision,
indicating that the relaxation becomes effectively tight in this regime. Fig.~\ref{fig:gapN_short} shows the Pareto
gap $\Delta(N)$, which collapses to (numerical) zero beyond $N_{\mathrm{th}}^{(\mathrm{emp})}$.

\section{Conclusion}

We studied privacy mechanism design under a \emph{sparse point-wise leakage} criterion that enforces \emph{per-output}
control of disclosure by limiting (i) the number of sensitive realizations that any released symbol can influence and
(ii) the corresponding worst-case leakage toward them. In the high-privacy regime (small $\epsilon$), an
information-geometric local expansion yields a quadratic approximation of the privacy--utility trade-off. When
$P_{X|Y}$ is invertible, the resulting approximate design reduces to an $\ell_0$-constrained Rayleigh-quotient
maximization over the subspace $\sqrt{P_X}^{\perp}$; moreover, it is without loss of optimality to restrict to a binary
released variable $U$ with uniform mass. The induced sparse Rayleigh-quotient is NP-hard in general (via its connection
to sparse PCA), which motivates our semidefinite relaxation and rounding procedure. Finally, we characterized a
deterministic sparsity threshold beyond which the sparse optimum saturates at the unconstrained spectral value and the
SDP relaxation becomes exact (rank-one), explaining the sharp transition observed in the numerical results.
\clearpage   
\bibliographystyle{IEEEtran}
\bibliography{IEEEabrv,IZS}

\end{document}